\documentclass{sig-alternate-10pt}
\usepackage{graphicx}
\usepackage{caption}
\usepackage{subcaption}
\usepackage{xcolor}
\usepackage[normalem]{ulem}
\usepackage{amsmath}
\newtheorem{theorem}{Theorem}[section]
\newtheorem{lemma}[theorem]{Lemma}

\newtheorem{definition}[theorem]{Definition}      

\title{Network Coding Applications for 5G Millimeter-Wave
  Communications} \author{Murali Narasimha \hspace{2cm} Hossein
  Bagheri\\ \\ Motorola Mobility Limited}

\begin{document}
\bibliographystyle{plain}
\maketitle
\begin{abstract}
The millimeter-wave bands have been attracting significant interest as
a means to achieve major improvements in data rates and network
efficiencies. One significant limitation for use of the
millimeter-wave bands for cellular communication is that the
communication suffers from much higher path-loss compared to the
microwave bands. Mil\linebreak[0]limeter-wave links have also been
shown to change rapidly, causing links between devices and access
points to switch among line-of-sight, non-line-of-sight and outage
states.  We propose using Random Linear Network Coding to overcome the
unreliability of individual communication links in such
millimeter-wave systems. Our system consists of devices transmitting
and receiving network-coded packets through several access points. For
downlink communication, network-coded packets are transmitted to a
device via multiple access points. For uplink communication, the
access points perform network coding of packets of several devices,
and then send the network-coded packets to the core network.  We
compare our approach against a naive approach in which non-network
coded packets are transmitted via multiple access points
(``Forwarding'' approach). We find that the network coding approach
significantly outperforms the forwarding approach. In particular,
network coding greatly improves the efficiency of the air-interface
transmissions and the efficiency of the backhaul transmissions for the
downlink and uplink communication, respectively.
\end{abstract}

\section{INTRODUCTION}
The mobile communications industry is embarking on a wide-ranging
effort to define and build the next generation of wireless systems
referred to as 5G. This is driven by large increases in mobile
communications usage and expected future growth that is staggering
\cite{Cisco2015},\cite{UMTS Forum2011}. 5G networks will be
expected to deliver as much as 1000 times the data rate of current
networks \cite{Osseiran2014}, \cite{Boccardi2014}.

As part of the 5G initiatives, there is enormous interest in
millimeter-wave (MmWave) bands between 30 GHz and 300 GHz
\cite{Boccardi2014}-\cite{Murdock2014}, where the available bandwidths
are much larger than today's cellular bands. While millimeter-wave has
historically been used for backhaul links and satellite
communications, it has not been considered suitable for cellular
communications due to the much higher path-loss that mmWave signals
experience. In order to overcome the path-loss, it is necessary to use
antenna arrays and perform beam-forming \cite{Khan2011}. Advances in
RF and semiconductor technologies have made the use of mmWave bands
more suitable for cellular communications \cite{Gutierrez2011}.
Specifically, such advances have made it possible to have antenna
arrays in areas small enough that they can be practically accommodated
inside mobile devices. Such beam-forming is considered an essential
enabling technology for millimeter-wave communication.

One distinctive characteristic of mmWave communication compared to
microwave is the highly directional nature of the signal path
\cite{Rangan2014}, \cite{Shokri-Ghadikolaei2015}. The beam-formed
communication makes the communication links much more directionally
sensitive. Obstructions can easily block the communication path; small
changes in orientation or small movements may cause \emph{``deafness''
}(due to the transmitter and receiver antennas not being correctly
pointed at each other) \cite{Talukdar2014}. Events that cause such
blocking and deafness are highly unpredictable (mobility, small
movements by user, temporary vehicular and other obstructions in the
environment, etc.). This makes it necessary to have redundant access
points such that in the event of a loss of communication to an access
point, the device can quickly switch to a different access point
\cite{Shokri-Ghadikolaei2015}.  However, switching to different access
points can result in interruptions in communication, with the overall
impact of the interruptions dependent on the frequency of the
switching and the time the device takes to establish an alternate
communication link.

Another distinctive characteristic is that the Doppler spread at
mmWave frequencies is much higher for a given speed, compared to
microwave frequencies. This leads to much smaller channel coherence
times compared to microwave (for example, at 30 km/h the coherence
times for 3 GHz and 30 GHz are 12 ms and 1.2 ms respectively),
suggesting that the channel can change very rapidly. Scheduling in
such an environment is likely to require very frequent feedback, which
tends to consume a lot of resources and energy. Moreover, if the
feedback communication occurs in the mmWave band, it is subject to the
same link breakage constraints mentioned above.

Relying on the redundancy of access points, we consider an
architecture in which a device maintains communication links to a
number of access points. Data can be transmitted from any of the
access points to the device and from the device to any of the access
points. We use network-coding techniques for data transmission on both
the downlink and the uplink. The general idea is that such an approach
is less dependent on frequent feedback, and the device continues to
receive (or transmit) the same data stream through some sub-set of the
access points even as connections to individual access points are
interrupted.

The paper is organized as follows. Section~\ref{sec:RelWork}
summarizes some of the related work in the area of millimeter wave
communications and network coding. Section~\ref{sec:SysMod} introduces
the system model and describes our use of Random Linear Network
Coding. We describe metrics that enable comparison of the network
coding based transmissions to conventional transmissions schemes. We
also derive bounds and estimates for the comparison metrics.
Section~\ref{sec:Sim} provides results of a Monte-Carlo system
simulation performed to understand the benefits of network coding in a
millimeter-wave deployment with access points deployed along a
street. Section~\ref{sec:Concl} has our concluding
  remarks and observations on possible future work.

\section{RELATED WORK}
\label{sec:RelWork}
\begin{figure}
\fbox{\includegraphics[height=4.1cm]{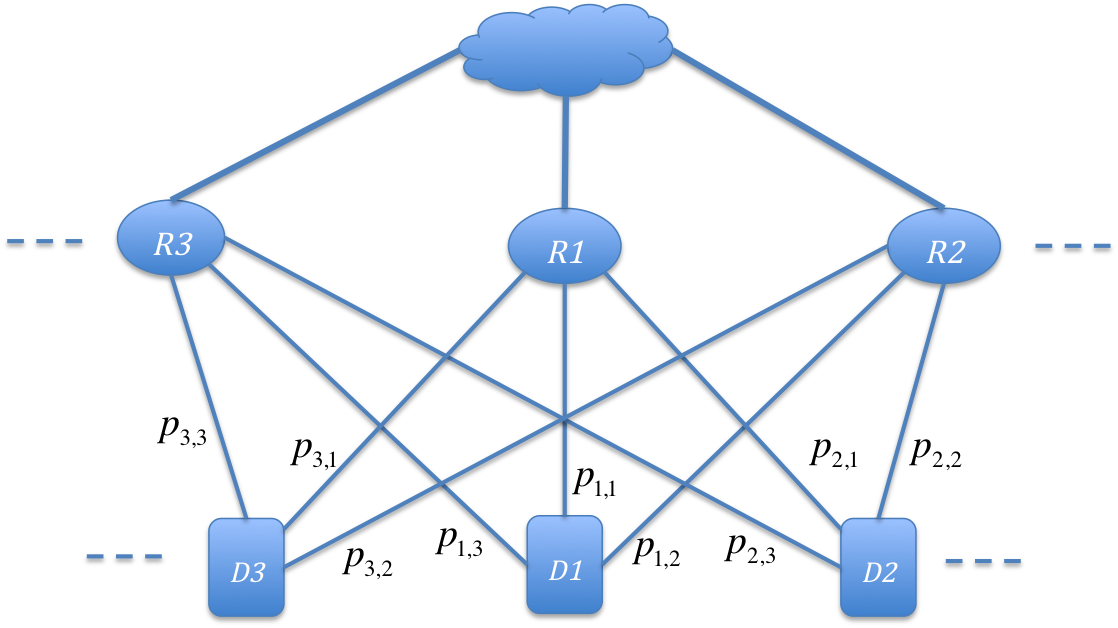}}
\caption{Block diagram of a network consisting of multiple Relay Nodes
  and multiple devices}
\label{fig:Fig1}
\end{figure}

Millimeter-wave for cellular networks has received a lot of attention
recently \cite{Boccardi2014}-\cite{special_issue1_2014}.  The issues
related to directionality of transmissions and the need for frequent
handovers have been studied in \cite{Talukdar2014}. Deployments of
Millimeter-wave access points for typical scenarios (street-side,
stadium, etc.) are considered and it is shown that, due to the
directionality and blocking properties, handovers occur much more
frequently than microwave networks.

Current cellular networks generally rely on the principle of having a
single serving access point. However, some enhancements, which involve
devices maintaining simultaneous links to different access points,
have been standardized. For example, Carrier Aggregation in LTE \cite{36_300}
allows for a device to be connected to and communicate with multiple
spatially separated nodes (Coordinated Multi-point Transmission).
Another example is Dual Connectivity in LTE \cite{36_300}, which allows for
concurrent communication with different (spatially separated) access
points, each independently scheduling packets for the device.

Network Coding was first proposed in \cite{Ahlswede2000} as a methodology for
improving throughput for multicast applications by transmitting
``combinations'' of packets. There have been some studies of
applications of network coding to wireless communications
\cite{Kishore2006}-\cite{Hausl2006}. To the best of our knowledge, there has not been any prior
work on investigating the applicability of network coding to
millimeter wave cellular networks.

\section{SYSTEM MODEL}
\label{sec:SysMod}
\begin{figure*}[t]
  \begin{tabular}{ll}
    {\begin{subfigure}[t]{0.47\textwidth}
        \fbox{\includegraphics[height=7cm, width=\textwidth]{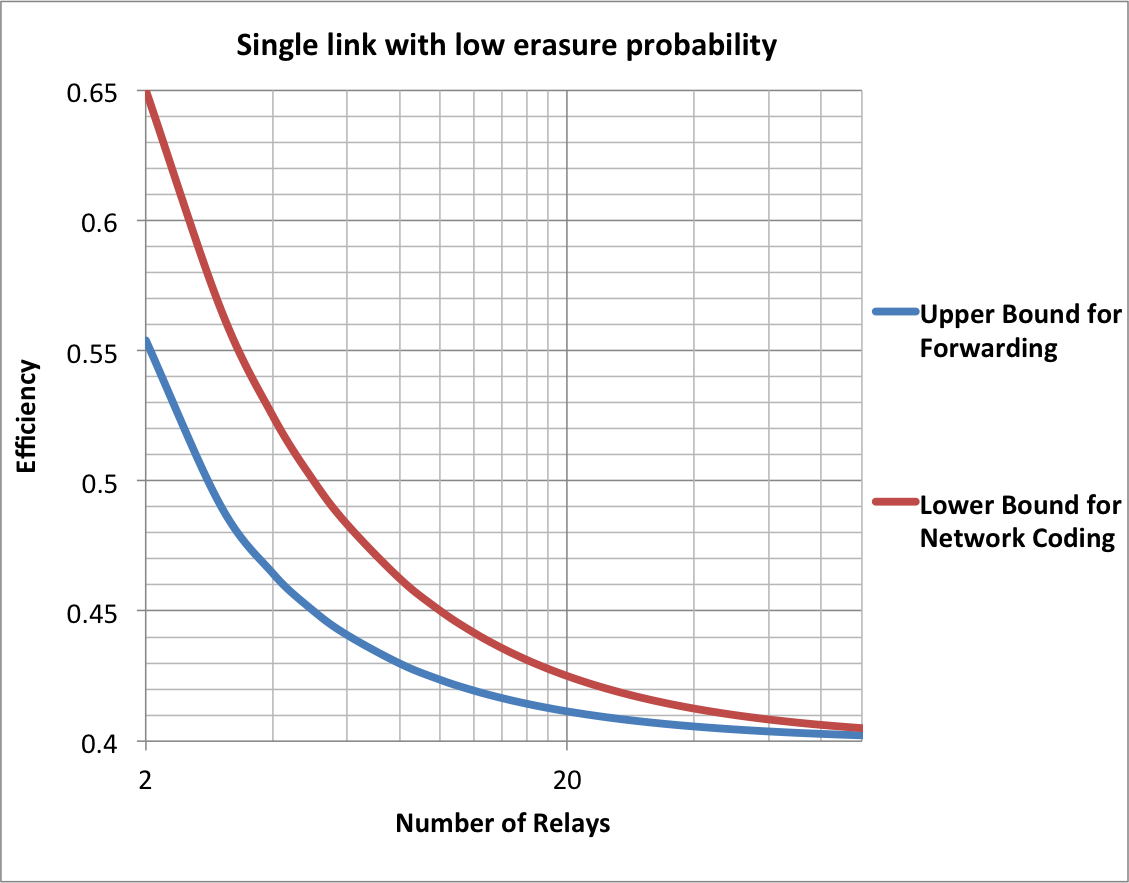}}
        \caption{\label{fig:Fig2a}}
      \end{subfigure}
    } &

    {\begin{subfigure}[t]{0.47\textwidth}
        \fbox{\includegraphics[height=7cm, width=\textwidth]{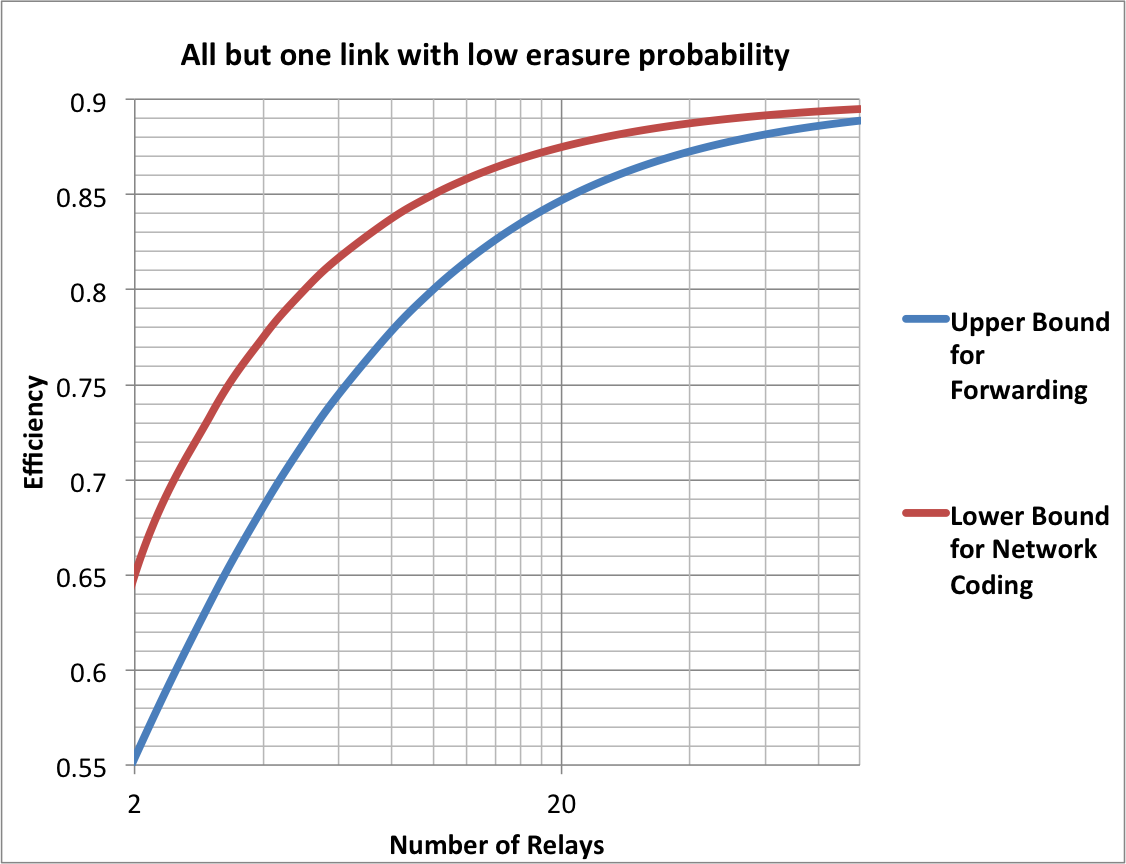}}
        \caption{\label{fig:Fig2b}}
      \end{subfigure}
    }
  \end{tabular}
  \caption{Upper and Lower Bounds for Forwarding and Network Coding
    respectively for Downlink transmissions: (\subref{fig:Fig2a}) for
    the case where there is a single link with low packet erasure
    probability, and (\subref{fig:Fig2b}) for the case where all but
    one link has low packet erasure probability. Low erasure
    probability of 0.1 and High erasure probability of 0.6.}
\label{fig:Fig2}
\end{figure*}

\begin{figure*}
  \begin{tabular}{ll}
    {\begin{subfigure}[t]{0.47\textwidth}
        \fbox{\includegraphics[height=7cm, width=\textwidth]{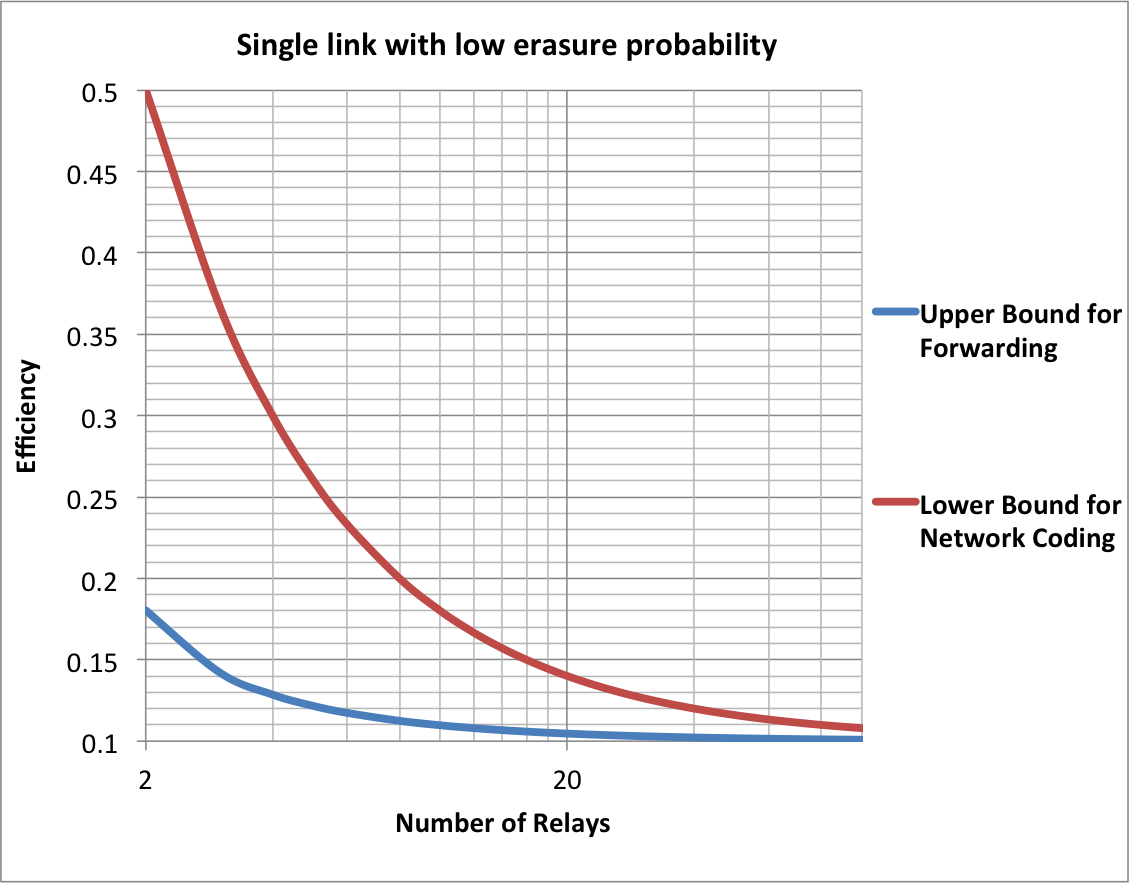}}
        \caption{\label{fig:Fig3a}}
      \end{subfigure}
    } &

    {\begin{subfigure}[t]{0.47\textwidth}
        \fbox{\includegraphics[height=7cm, width=\textwidth]{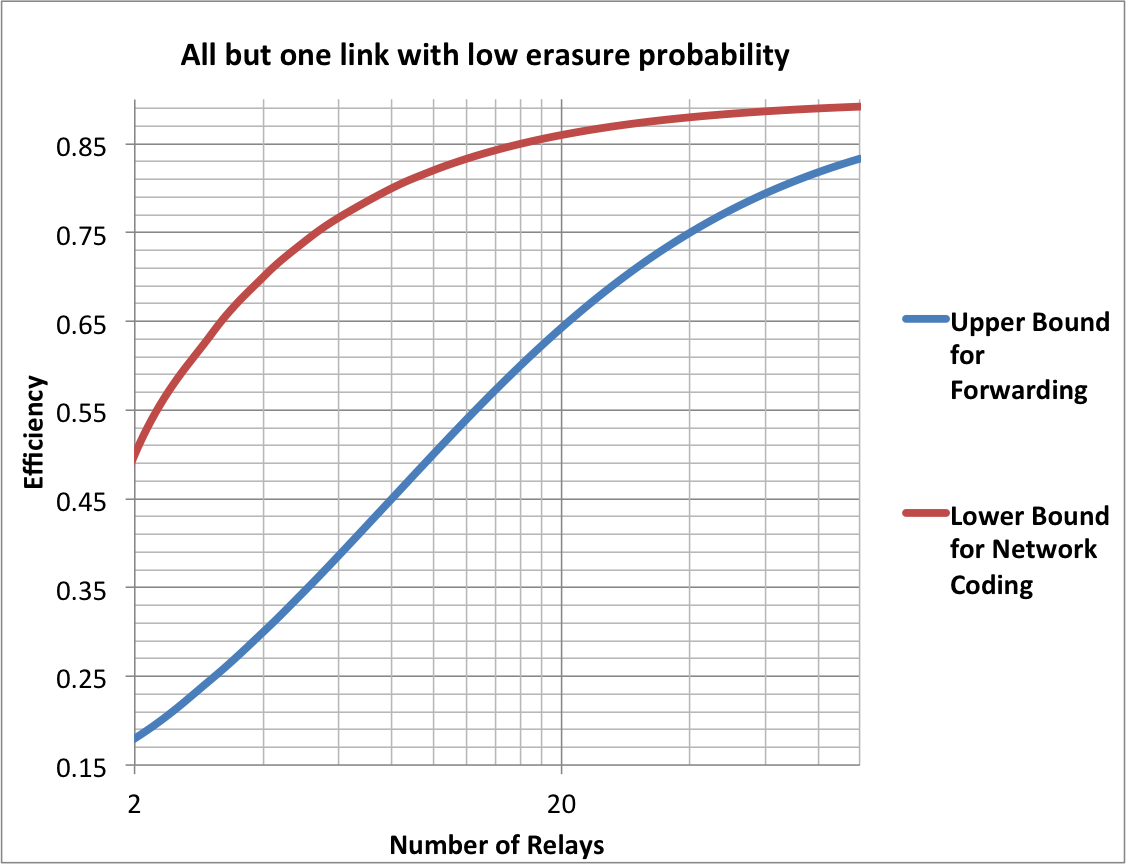}}
        \caption{\label{fig:Fig3b}}
      \end{subfigure}
    }
  \end{tabular}
  \caption{Upper and Lower Bounds for Forwarding and Network Coding
    respectively for Downlink transmissions: (\subref{fig:Fig3a}) for
    the case where there is a single link with low packet erasure
    probability, and (\subref{fig:Fig3b}) for the case where all but
    one link has low packet erasure probability. Low erasure
    probability of 0.1 and High erasure probability of 0.9.}
\label{fig:Fig3}
\end{figure*}

We consider a radio network comprising multiple relay nodes
$R1,R2,\ldots ,RN$ serving one or more devices $D1,D2,\ldots ,DN$ as
shown in Figure \ref{fig:Fig1}. Each device maintains mmWave
communication links to a set of relay nodes and data is concurrently
transmitted to and received from the device through the set of relay
nodes. The relay nodes are connected to the network through backhaul
links, which could be wired, or wireless.  Thus, the physical
realization of the relay nodes described here could be in the form of
small cells, wireless relays or remote radio heads.

Links between the devices and relay nodes are modeled as packet
erasure channels; $p_{i,j}$ represents the packet erasure probability
of the millimeter-wave link between device $Di$ and relay node $Rj$.
The packet erasure probabilities are time varying and their modeling
is described further below.

Millimeter-wave channel conditions can vary rapidly as shown in
[19]. In particular, a link can quickly change among line-of-sight,
non-line-of-sight and outage states.  Our study implicitly assumes
that the channel state is not known with certainty before transmitting
a packet to the device. Our intention is to compare performance of the
network coding based data transmission via several relay nodes to a
transmission approach in which non-network-coded packets are
transmitted via several relay nodes. We refer to the latter approach
as ``Forwarding''.  In this section, we first describe Random Linear
Network Coding, after which we describe the scheduling mechanisms for
the downlink and the uplink communication and provide some performance
bounds.

\subsection{Random Linear Network Coding}
\label{subsec-RLNC}
Given a set of $k$ packets $[P_1,P_2,\ldots ,P_k]$ chosen from a
Galois Field alphabet $F$, a random linear network-coded packet is
constructed as $\sum^{k}_{i=1} c_i\cdot P_i$ where $[c_1,c_2,\ldots
  ,c_k]$ is an ``encoding vector'' consisting of a set of coefficients
randomly chosen from $F$. If a receiver receives $k$ packets
$[r_1,r_2,\ldots ,r_k]$ with encoding vectors $[c^i_1,c^i_2,\ldots
  ,c^i_k]$ corresponding to each $r_i$, it can construct a transfer
matrix:
\[M = \left( \begin{array}{ccc}
    c^1_1 & c^1_2 & \cdots c^1_k \\
    c^2_1 & c^2_2 & \cdots c^2_k \\
    & \vdots \\
    c^k_1 & c^k_2 & \cdots c^k_k \end{array}
    \right)
\]
The receiver can then recover the original packets
using $M^{-1}\cdot[r_1,r_2,\cdots ,r_k]^T$. The
probability of $M$ not having an inverse can be made sufficiently small by choosing a large $F$.

\subsection{Scheduling and Data Transmission}
\label{subsec-sched}

In this study, time is divided into time-spans, with each time span
having multiple time-slots. A time-slot is the time duration of one
transmission over the air-interface.  A number of packets $k$ are to
be delivered in each time-span. In each time span the transmitting
side attempts to deliver the $k$ packets to the receiving side through
relays. The Forwarding approach and the Network coding approach are
described as follows:

\begin{itemize}
\item{\bf{\em{Forwarding}}}: For each of the $k$ packets, a relay node is
  selected. The transmitting side then performs transmissions until
  the packet is received.

\item{\bf{\em{Network Coding}}}: For downlink communication, while the
  $k$ packets are not decoded, the transmitting side (a) generates a
  Random Linear Network Coded (RLNC) packet from the $k$ original
  packets, (b) chooses a relay node, and (c) transmits the RLNC packet
  through the chosen relay node. For the uplink communication, any
  relay node that has received a subset of packets forms a network
  coded packet from the subset and transmits it to the network.
\end{itemize}

We compare the efficiency of the two approaches, which we define as
the ratio of the number of packets delivered to the number of
transmissions needed to deliver the packets. Estimating other metrics
such as throughput and delay requires specific assumptions regarding
the physical layer communication (e.g., transmission slot duration,
bandwidth, etc.). Although we do not model such specific physical
layer parameters, the efficiency metrics used are directly related to
throughput and delay. For example, a higher efficiency translates to a
higher data rate. A higher efficiency can also translate to lower
average packet delivery delay.

Comparing Network Coding to the Forwarding approach as described above
is motivated by having an ``apples-to-apples'' comparison. That is,
since the network coding approach uses multiple relay nodes, we
compare it to an approach that does not use network coding, but still
utilizes multiple relay nodes. Abstracting away from the physical
layer details, the Forwarding approach is similar to Coordinated
Multi-point Transmission in LTE ~\cite{Clerckx2012}. In the following
we study two types of network coding, namely: \emph{Intra-session
  network coding} for downlink communication, and \emph{Inter-session
  network coding} for uplink communication.

\subsubsection{Single Device - Intra-session Network Coding for Downlink}
\label{subsec-intrasessNC}

\begin{figure*}[t]
  \begin{center}
    \fbox{\includegraphics[width=14cm]{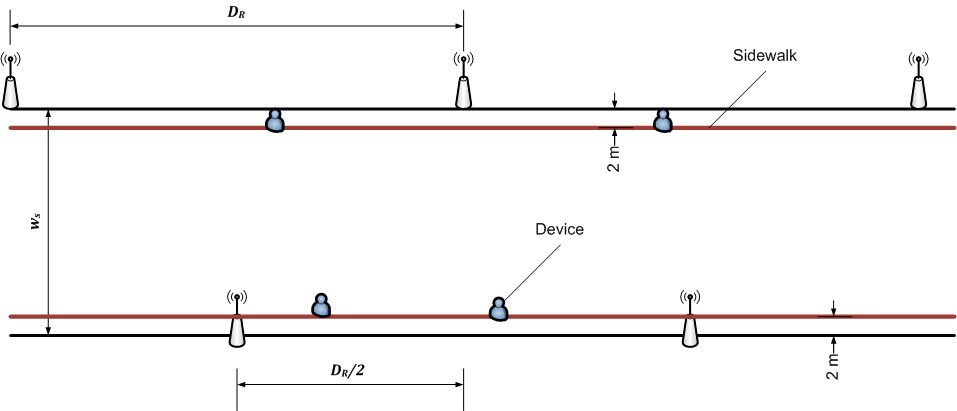}}
    \caption{Deployment scenario for simulation: Users are randomly
      dropped on sidewalks, and relays are equally spaced on each side of
      the street}
    \label{fig:Fig4}
  \end{center}
\end{figure*}

We first consider a setup in which the network is transmitting data to
a single device $D$ through $N$ relay nodes. Forwarding and Network
Coding are performed as described above and the selection of relays is
based on a uniform random distribution. To derive the efficiency of
the two approaches, let $R_1,R_2,\cdots ,R_N$ be the relay nodes,
$P_1,P_2,\cdots ,P_k$ be the packets to be transmitted in the
time-span, and $p_1,p_2,\cdots ,p_N$ be the packet erasure
probabilities of the links from the corresponding relay nodes to $D$.

For the forwarding approach, given the random selection of relay
nodes, each relay node is expected to transmit $\lceil \frac{k}{N}
\rceil$ original packets on average. The expected value of the number
of transmissions through relay node $R_i$ is $\lceil \frac{k}{N}\rceil
\frac{1}{(1-p_i)}$. The total number of transmissions to deliver the
$k$ packets is $\sum^{N}_{i=1} \lceil \frac{k}{N}\rceil
\frac{1}{(1-p_i)}$. The efficiency for the Forwarding approach is
therefore:

\begin{equation}
  \mathit{eff_F} = \frac{k}{\lceil \frac{k}{N}\rceil
    \sum^{N}_{i=1}\frac{1}{(1-p_i)}}
  \label{eq1}
\end{equation}

Observing that $\lceil\frac{k}{N}\rceil \geq \frac{k}{N}$, we derive
an \emph{upper bound} for the forwarding efficiency:

\begin{equation}
  \mathit{eff_F} \leq \frac{N}{\sum^{N}_{i=1}\frac{1}{(1-p_i)}}
\label{eq2}
\end{equation}

For the Network Coding approach, let
$\mathcal{P}_1,\mathcal{P}_2,\cdots ,\mathcal{P}_L$ be the RLNC
packets that are transmitted through the relay nodes to deliver the
original packets $P_1,P_2,\cdots ,P_k$. Given the random selection of
relay nodes, each relay node is expected to transmit $\lceil
\frac{L}{N}\rceil$ RLNC packets on average. Of the number of RLNC
packets transmitted through a given relay node, the number that is
successfully received follows a binomial distribution with parameters
$\lceil\frac{L}{N}\rceil$ and $p_i$. The expected value of the number
of RLNC packets received through $R_i$ is
$\lceil\frac{L}{N}\rceil(1-p_i)$.  The total number of RLNC packets
transmitted to deliver the $k$ original packets is
$\sum^{N}_{i}\lceil\frac{L}{N}\rceil(1-p_i)$ . In order to decode the
$k$ original packets, $k$ RLNC packets need to be received. Therefore
we have:
\begin{equation}
  k = \lceil{\frac{L}{N}}\rceil \sum^{N}_{i=1}(1-p_i)
  \label{eq3}
\end{equation}
The efficiency of the network coding approach is therefore:
\begin{equation}
  \mathit{eff_C} = \frac{\lceil\frac{L}{N}\rceil \sum^{N}_{i=1}(1-p_i)}{L}
\label{eq4}
\end{equation}

Noting that $\lceil\frac{L}{N}\rceil \geq \frac{L}{N}$, we derive the
following \emph{lower bound} for network coding:
\begin{equation}
  \mathit{eff_C} \geq \frac{\sum^{N}_{i=1} (1-p_i)}{N}
\label{eq5}
\end{equation}

Note that the difference between the forwarding and the network coding
approaches is that for the forwarding case, a single node is
responsible for delivering a particular packet leading to bound
(\ref{eq2}) whereas for the network coding case, it is sufficient to
correctly receive a total of $k$ packets that may have been received
via multiple nodes.

To understand how the two schemes perform compared to each other, the
best (i.e., upper bound) performance of forwarding is compared to the
worst (i.e., lower bound) performance of network coding for three
different scenarios described in the following. In the first scenario,
$p_{i}=p, \ \forall i$. In this scenario, we get $\mathit{eff_F}\leq
1-p \leq \mathit{eff_C}$. That is, when all links have equal erasure
probability, network coding performs at least as well as forwarding.

Figures \ref{fig:Fig2} and \ref{fig:Fig3} show a comparison of the
efficiency bounds for a second and a third specific scenario
respectively. For the second scenario, Figure \ref{fig:Fig2a}, one
link from $D$ to one of the Relays has a low packet erasure
probability, while all the other links to the Relays have a high
packet erasure probability. For the third scenario, Figure
\ref{fig:Fig2b}, one link from $D$ to one of the relays has a high
packet erasure probability, while all the other links have a low
packet erasure probability. For both asymmetric erasure probability
cases a low and high packet erasure probabilities of 0.1 and 0.6
respectively are assumed. Similarly, Figure \ref{fig:Fig3} shows a
corresponding comparison of the efficiency bounds for asymmetric cases
with low and high packet erasure probabilities of 0.1 and 0.9
respectively. That is, Figure \ref{fig:Fig3a} shows the upper bound
for forwarding efficiency and the lower bound for network coding
efficiency with a single link having low erasure probability and all
other links having high erasure probability. Figure \ref{fig:Fig3b}
shows the upper bound for forwarding efficiency and the lower bound
for network coding efficiency with a single link having high erasure
probability and all other links having low erasure probabilities.

The main observation from Figures \ref{fig:Fig2} and \ref{fig:Fig3} is
that the performance improvements that are expected for network coding
relative to forwarding are larger for scenarios where there is a larger
difference between the erasure probabilities on the different links.
It is also worth emphasizing that Figure \ref{fig:Fig2} and Figure
\ref{fig:Fig3} show upper and lower bounds for the forwarding
efficiency and the network coding efficiency, respectively. That is,
under the assumptions made above, \emph{forwarding is expected to
  perform no better than and network coding is expected to perform at
  least as well as the efficiency plots shown}.

\subsubsection{Multiple Devices - Inter-session Network Coding for Uplink}
\label{subsec-intersess}

We now consider a setup in which devices \(D1,D2,\cdots \linebreak[0]
,Dz\) transmit to the network through Relay nodes
\(R_1,R_2,\linebreak[0] \cdots ,R_N\). The goal is to compare the
efficiency of the backhaul usage for forwarding and network coding.

Each device transmits one packet in a time span, with re-transmissions
as necessary to ensure that at least one of the relay nodes receives
the packet. For the forwarding approach, each Relay node simply
transmits to the network the packets that it has received\footnote{For
  both Forwarding and Network coding, we assume that a given Relay
  node has no knowledge of which packets have been received by other
  Relay nodes. Such knowledge would require extensive signaling
  between the Relay nodes.}. For the network coding approach, each
Relay node constructs network-coded packets from the original packets
it has received and transmits them to the network over the backhaul
link. Specifically, if a relay node does not receive a packet from a
particular device, it uses an encoding coefficient of 0 corresponding
to that device; otherwise, a non-zero encoding coefficient is randomly
chosen. Once an adequate number of packets are received to be able to
decode the original packets at the network side in the time span, the
relay nodes do not transmit further network-coded packets during that
time span.

To calculate the backhaul efficiency, let $P_1,P_2,\cdots ,P_z$ be the
packets from \(D1,D2,\cdots ,Dz\) in a time span. Let
$p_{i,1},p_{i,2},\cdots ,p_{i,N}$ be the packet erasure probabilities
for the links from $Di$ to $R_1,R_2,\cdots ,R_N$.  As mentioned above,
each device $Di$ transmits its packet $P_i$ until it is received by at
least one of the relay nodes. The expected number of transmissions of
$P_i$ can be written as: 
\begin{equation}
\mathbb{E}[[P_i]] = \frac{1}{1 - \prod^N_{j=1}p_{i,j}}.
\end{equation}

The probability distribution of the number of times $r_{i,j}$ that
$P_i$ is received successfully at $R_j$ from $n_{i}$ transmissions
of $P_i$, follows a binomial distribution. Thus the probability that
$P_i$ is received successfully at $R_j$ at least once is
$1-{p_{i,j}}^{n_{i}}$. We approximate the probability of $R_j$
receiving $P_i$ as $1-(p_{i,j})^{\mathbb{E}[[P_i]]}$.  For forwarding,
given that the Relays simply transmit the packets received from the
devices, the expected value of the number of backhaul transmissions of
$P_i$ is obtained by summing over the relay nodes, i.e.,
$\sum^{N}_{j=1}(1-(p_{i,j})^{\mathbb{E}[[P_i]]})$.  The expected value
of the of the number of backhaul transmissions of all packets
$P_1,P_2,\cdots ,P_z$ is given by:
\begin{equation}
  \sum^{z}_{i=1} {\sum^{N}_{j=1}(1-(p_{i,j})^{\mathbb{E}[[P_i]]})}
\end{equation}
assuming no acknowledgment feedback from the network to relays for
each individual packet. The backhaul efficiency is defined as the
ratio of the number of original packets per time span to the average
number of backhaul transmissions per time span. The expected value of
the backhaul efficiency is approximated as:
\begin{equation}
  \mathit{bkEff_F} \approx \frac{z}{\sum^{z}_{i=1}
    {\sum^{N}_{j=1}(1-(p_{i,j})^{\mathbb{E}[[P_i]]})}}
  \label{eq:bkeff_f1}
\end{equation}

For erasure probabilities that are close to 1, the approximation of Eq
(\ref{eq:bkeff_f1}) can yield efficiency values that are greater than
1. This is due to the fact that the expected value of the number of
backhaul transmissions of $P_i$, as represented by $\sum_{j=1}^N (1 -
(p_{i,j})^{\mathbb{E}[[P_i]]})$, is less than 1 for cases where $N$ is
small and $p_{i,j}$ are close to 1. However, since the number of
backhaul transmissions of $P_i$ is at least 1 in practice, we modify
Eq (\ref{eq:bkeff_f1}) as follows\footnote{It is remarked that we
  found the above approximation to be quite tight through numerous
  Monte-Carlo simulations, particularly for low to medium erasure probabilities.}:
\begin{equation}
  \mathit{bkEff_F} \approx \frac{z}{\sum^{z}_{i=1}
    \max(1, {\sum^{N}_{j=1}(1-(p_{i,j})^{\mathbb{E}[[P_i]]})})}
  \label{eq:bkeff_f2}
\end{equation}

\begin{figure*}
  \begin{tabular}{ll}
    {\begin{subfigure}[t]{0.47\textwidth}
        \fbox{\includegraphics[height=7.5cm,width=\textwidth]{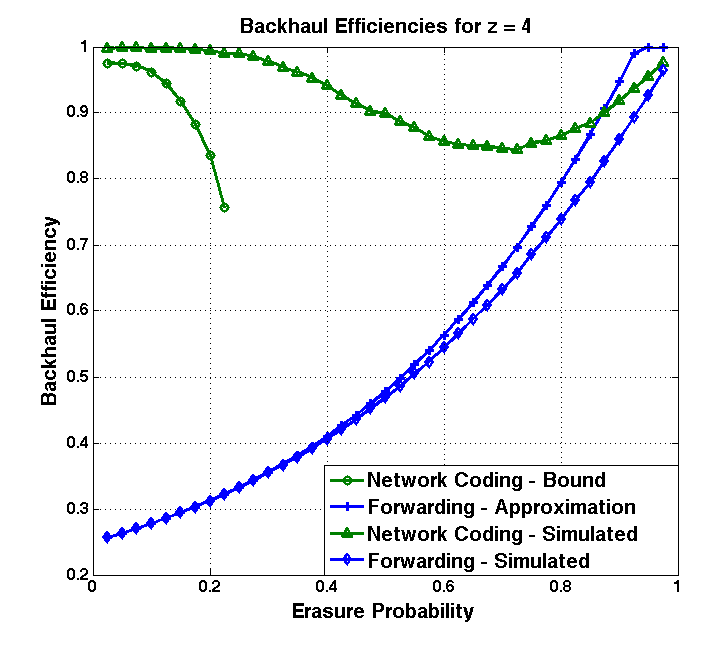}}
      \end{subfigure}
    } &

    {\begin{subfigure}[t]{0.47\textwidth}
        \fbox{\includegraphics[height=7.5cm,width=\textwidth]{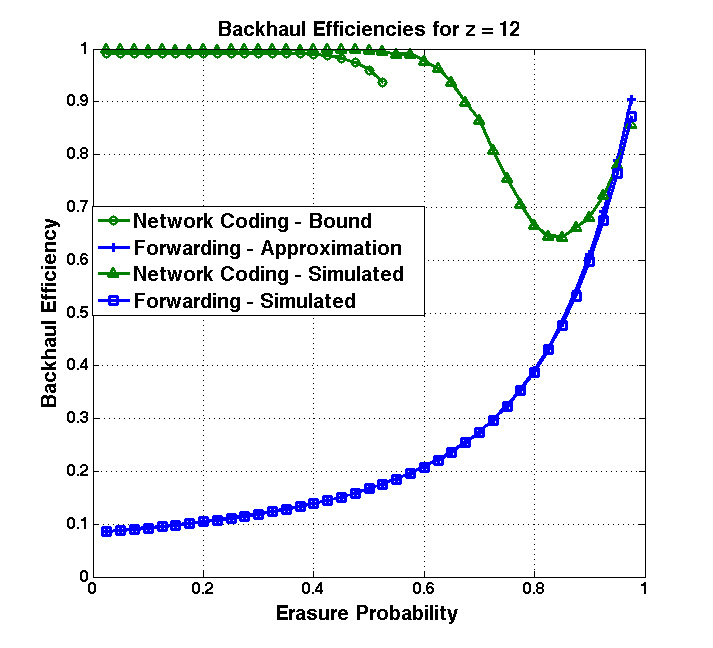}}
      \end{subfigure}
    }
  \end{tabular}
  \caption{Theoretical and Simulated Backhaul Efficiencies for equal
    link erasure probabilities}
  \label{fig:beffc_vs_befff}
\end{figure*}

For the Network Coding case, the backhaul efficiency calculation
involves subtleties in matrix rank calculation considering each relay
may have received a subset of packets to be network-coded.  In the
following, a lower bound on the efficiency is obtained for the
scenario in which $\forall i,j,\ p_{i,j}=p $, and $N\leq
z$,\footnote{The $N>z$ case is discussed at the end of this section.}
and the bound derivation for the \emph{asymmetric} erasure probability
cases is left as future work. The main steps of the derivations are
provided here, while the details are presented in Appendix \ref{sec:
  Appendix-NC-Backhaul}.
\begin{enumerate}
\item The average number of the required backhaul transmissions
  $\beta_{nc}$ is upper bounded in Eq (\ref{eq: beta_bound1}) (hence
  providing a lower bound for the backhaul efficiency of network coding).  The
  wireless network needs to collect enough (i.e., $\beta_{nc}$)
  network-coded packets to ensure $z$ of them are independent. The
  encoding coefficients corresponding to each received network-coded
  packet are randomly drawn from a Galois field of size $q$ according
  to the distribution $\Omega$: $0$ with probability $p$, and $w \in \{
  {1,\cdots, q-1}\}$ with probability $\frac{1-p}{q-1}$. To upper bound
  $\beta_{nc}$, we assume the probability of a randomly chosen $z
  \times z$ matrix $M$ being singular is $\phi$, if each element of
  the matrix is drawn according to $\Omega$.
\begin{eqnarray}\label{eq: beta_bound1}
\beta_{nc}(\phi) \leq \sum_{l=0}^{\infty} (z+l)\,
\big[1-\phi^{\zeta(l)-\zeta(l-1)}\big] \,\phi^{\zeta(l-1)} ,
\end{eqnarray}
where $\zeta(l)=\binom{z+l}{z}$ and $\zeta(-1)=0$.
\item It is shown that for $\phi\leq \phi_{ub} <1$, $\beta_{nc}(\phi) <
  \beta_{nc}(\phi_{ub})$, where
\begin{equation}\label{eq: phi_ub1}
\phi_{ub}=\log_q(\bar{l}(M)+1),
  \end{equation}
  and $\bar{l}(M)$ denotes the expected number of linear dependencies
  of the rows of $M$ (see definition \ref{def: ld}).
\end{enumerate}

We derive a lower bound for the backhaul efficiency
for network coding shown below:
\begin{equation}\label{eq:beffc}
  \mathit{bkEff_C} \geq \frac{z}{\sum_{l=0}^{\infty} (z+l)\ 
    \big[1-\phi_{ub}^{\zeta(l)-\zeta(l-1)}\big] \ 
    \phi_{ub}^{\zeta(l-1)}}
\end{equation}
The details of the derivation are provided in appendix
\ref{sec-ncBoundDerv}. Note that the bound proposed in Eq
(\ref{eq:beffc}) is valid for $\phi<1$, which for a given code length
$z$ and field size $q$, translates to a feasible region for erasure
probability $p$ according to Eq. (\ref{eq: phi_ub1}).

Since the network coding bound is derived only for the case of
symmetric link erasure probability, we compare in this section the
backhaul efficiency for both Forwarding and Network Coding assuming
$p_{i,j}=p \ \forall i,j$, and leave the analysis of the general case
to Section \ref{subsec-UpSim}. When $\forall i,j,\ p_{i,j}=p$:
  \[\mathit{bkEff_F} \approx
  \min\Big[1,\frac{1}{N(1-p^{\frac{1}{1 - p^N}})}\Big].\] 

Figure \ref{fig:beffc_vs_befff} compares the following: (i) the lower
bound of network coding backhaul efficiency, (ii) the approximate
forwarding backhaul efficiency, (iii) the simulated backhaul
efficiency for network coding, and (iv) the simulated backhaul
efficiency for forwarding. Figure \ref{fig:beffc_vs_befff}(a) and
Figure \ref{fig:beffc_vs_befff}(b) show the comparison for $z = 4$ and
$z = 12$ respectively, and $q=1024$ is assumed for both. It can be
seen from the figure that:
\begin{itemize}
\item The network coding backhaul efficiency bound is applicable to a
  range of erasure probability $p$, and outside of this range,
  $\mathit{bkEff_C}$ is undefined. The range of $p$ for which
  $\mathit{bkEff_C}$ is defined corresponds to
  $0\leq\phi<1$. Additionally, increasing the code length $z$
  increases the range over which $\mathit{bkEff_C}$ is defined.
\item The network coding efficiency is close to 1 for
  a range of erasure probabilities. As code length $z$ increases the
  range of erasure probabilities for which the efficiency is close to
  1 grows.
\item For medium to high erasure probabilities, the network coding
  lower bound in Eq (\ref{eq:beffc}) deviates significantly from the
  simulated network coding backhaul efficiency. The following is an
  explanation for this deviation. Note that Eq (\ref{eq:beffc})
  does not take into account which relay nodes transmit the
  packets. Consequently, it includes cases where a relay node
  transmits redundant packets to the network (for example, a relay
  node may transmit a first network coded packet based on an encoding
  vector $[c_1, 0, 0, 0]$ and a later network coded packet based on an
  encoding vector $[c_2, 0, 0, 0]$, where $c_1$, and $c_2$ are
  randomly chosen coefficients from alphabets of the field). In the
  computation of the simulated network coding backhaul efficiencies,
  transmission of such redundant packets is avoided. The redundant
  transmissions cause a gap between the lower bound and the simulated
  backhaul efficiency, and the size of the gap increases with
  increasing erasure probabilities. Deriving a tighter lower bound for
  the network coding backhaul efficiency by avoiding transmission of
  redundant network coded packets is a topic for further study.
\end{itemize}
We make a final remark regarding the $N>z$ case. If relay nodes
transmit their respective network coded packets at the same time, it
is possible that one network coded packet from each relay node (i.e.,
$N$ network coded packets) is transmitted to the network although
fewer than $N$ network coded packets are needed to decode the $z$
packets. In such a case, the upper bound on the number of backhaul
transmissions would be $\max(N, \beta_{nc}(\phi_{ub}))$. Thus, if
$N>z$ and $p$ is small, there can be redundant transmissions, which in
turn reduces the backhaul efficiency for network coding. Consequently,
it is not beneficial to use more than $z$ relay nodes in such cases.

\section{SIMULATION}
\label{sec:Sim}
\subsection{Deployment Scenario}
\label{subsec-DepScn}
\begin{figure*}
  \begin{tabular}{cc}
    
    { \begin{subfigure}[t]{0.47\textwidth}
        \fbox{\includegraphics[height=5.8cm,width=\textwidth]{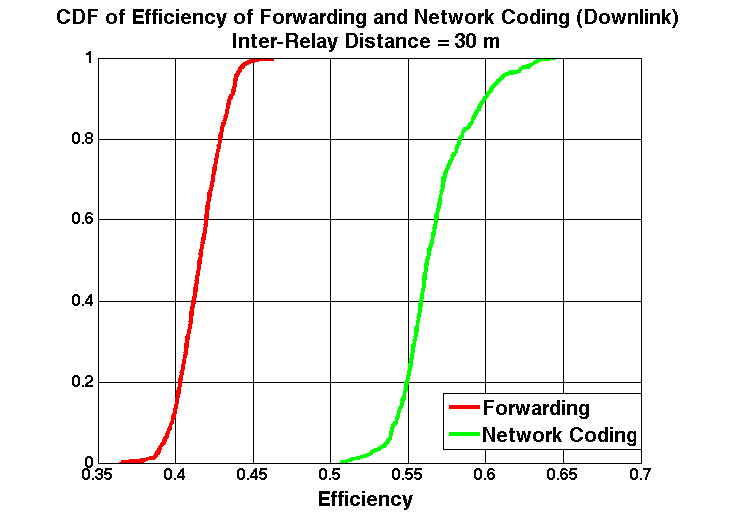}}
        \caption{\label{fig:Fig5a}}
      \end{subfigure}
    }
    &%
    { \begin{subfigure}[t]{0.47\textwidth}
        \fbox{\includegraphics[height=5.8cm,width=\textwidth]{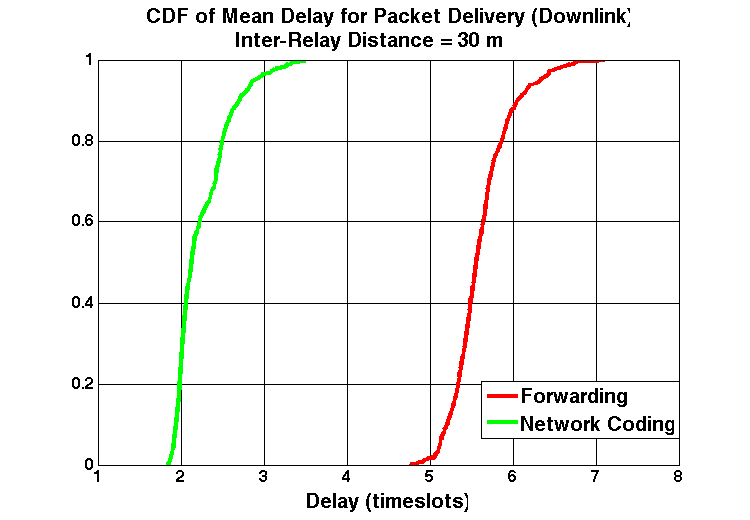}}
        \caption{\label{fig:Fig5b}}
      \end{subfigure}
    } \\

    { \begin{subfigure}[t]{0.47\textwidth}
        \fbox{\includegraphics[height=5.8cm, width=\textwidth]{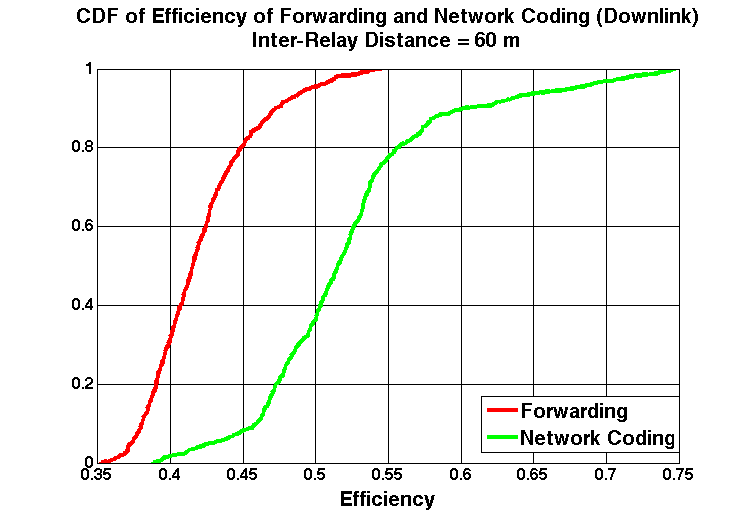}}
        \caption{\label{fig:Fig5c}}
      \end{subfigure}
    }
    &%
    { \begin{subfigure}[t]{0.47\textwidth}
        \fbox{\includegraphics[height=5.8cm, width=\textwidth]{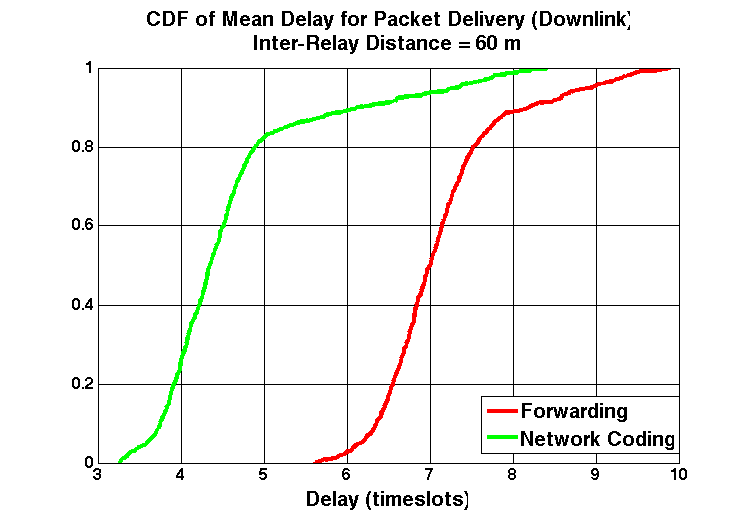}}
        \caption{\label{fig:Fig5d}}
      \end{subfigure}
    }\\

    { \begin{subfigure}[t]{0.47\textwidth}
        \fbox{\includegraphics[height=5.8cm, width=\textwidth]{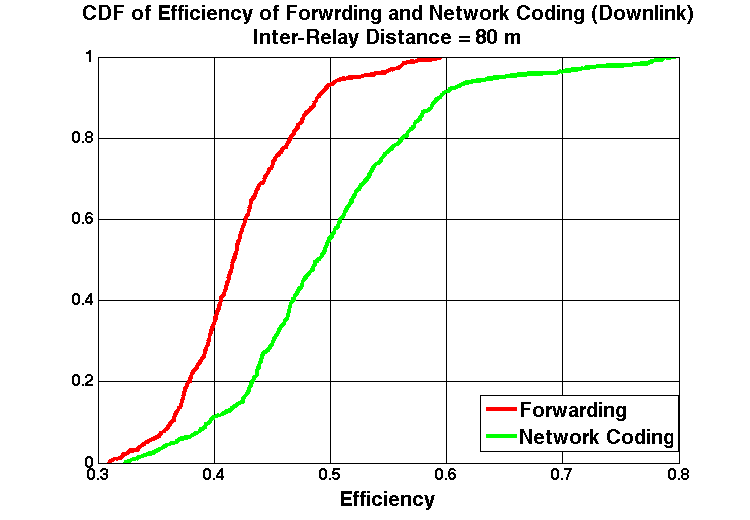}}
        \caption{\label{fig:Fig5e}}
      \end{subfigure}
    }
    &%
    { \begin{subfigure}[t]{0.47\textwidth}
        \fbox{\includegraphics[height=5.8cm, width=\textwidth]{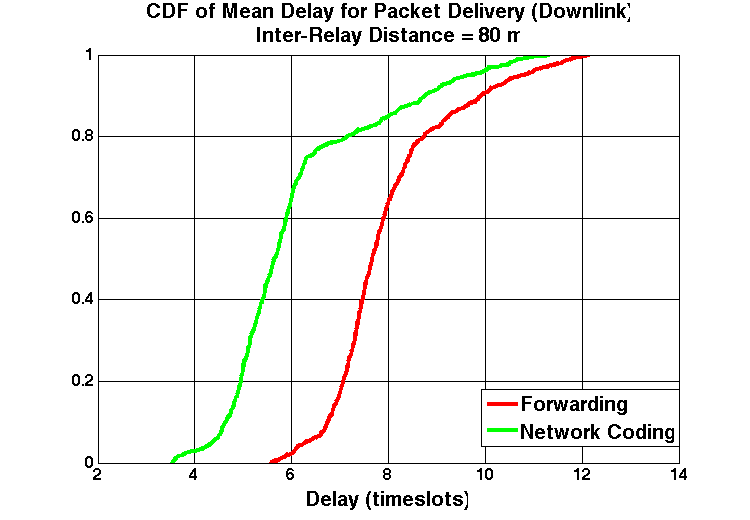}}
        \caption{\label{fig:Fig5f}}
      \end{subfigure}
    }
  \end{tabular}
  \caption{Downlink efficiency and delay performance of network coding
    compared to forwarding for different relay densities}
\label{fig:Fig5}
\end{figure*}

\begin{figure*}
  \begin{tabular}{cc}

    { \begin{subfigure}[t]{0.47\textwidth}
        \fbox{\includegraphics[height=5.8cm, width=\textwidth]{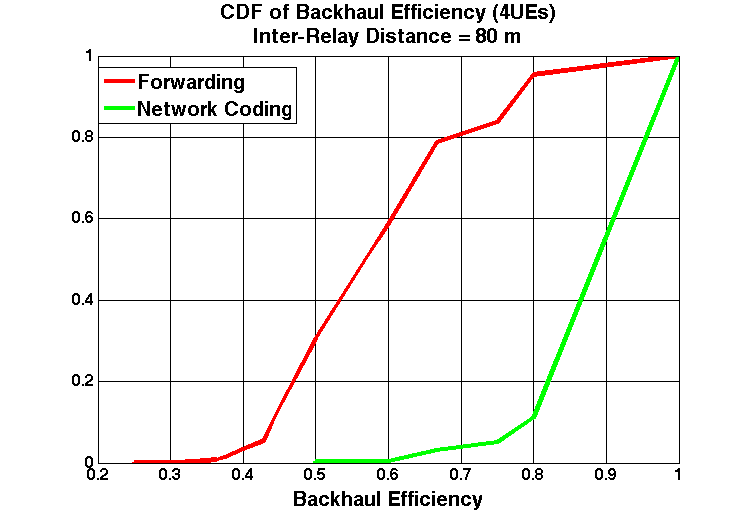}}
        \caption{\label{fig:Fig6a}}
      \end{subfigure}
    }
    &%
    { \begin{subfigure}[t]{0.47\textwidth}
        \fbox{\includegraphics[height=5.8cm, width=\textwidth]{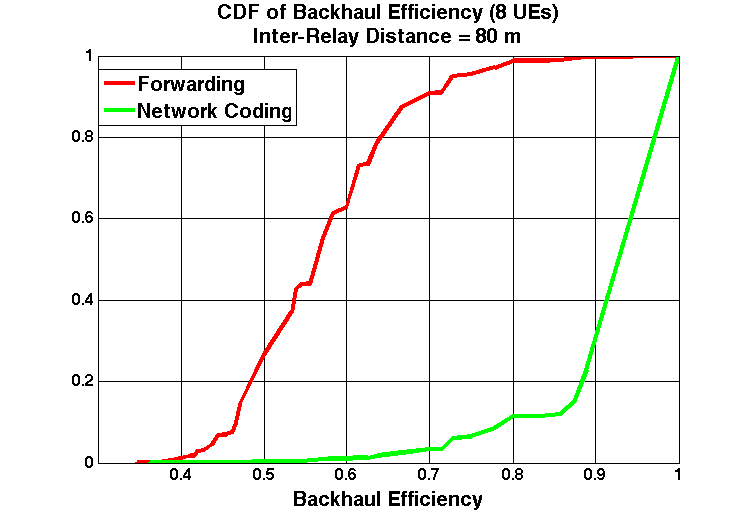}}
        \caption{\label{fig:Fig6b}}
      \end{subfigure}
    } \\%

    {\begin{subfigure}[t]{0.47\textwidth}
        \fbox{\includegraphics[height=5.8cm, width=\textwidth]{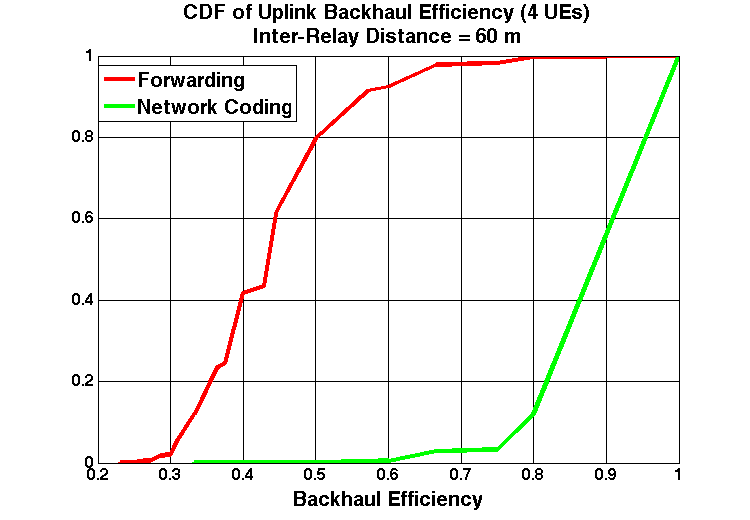}}
        \caption{\label{fig:Fig6c}}
      \end{subfigure}
    }
    &%
    {\begin{subfigure}[t]{0.47\textwidth}
        \fbox{\includegraphics[height=5.8cm, width=\textwidth]{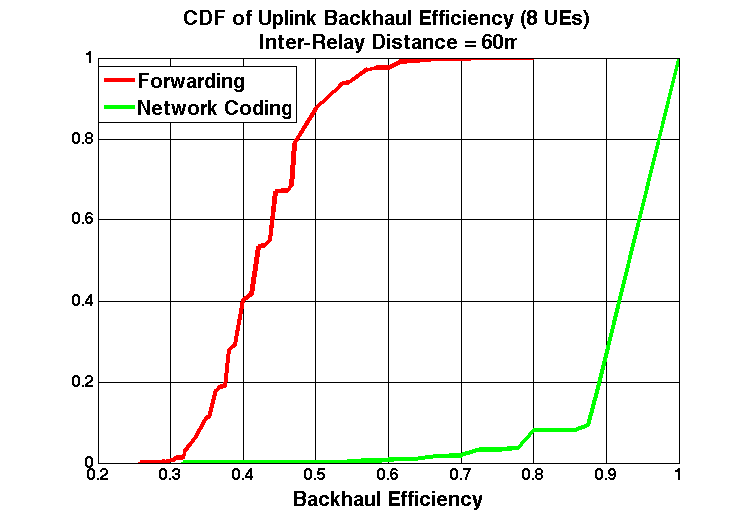}}
        \caption{\label{fig:Fig6d}}
      \end{subfigure}
    } \\%
      
    { \begin{subfigure}[t]{0.47\textwidth}
        \fbox{\includegraphics[height=5.8cm, width=\textwidth]{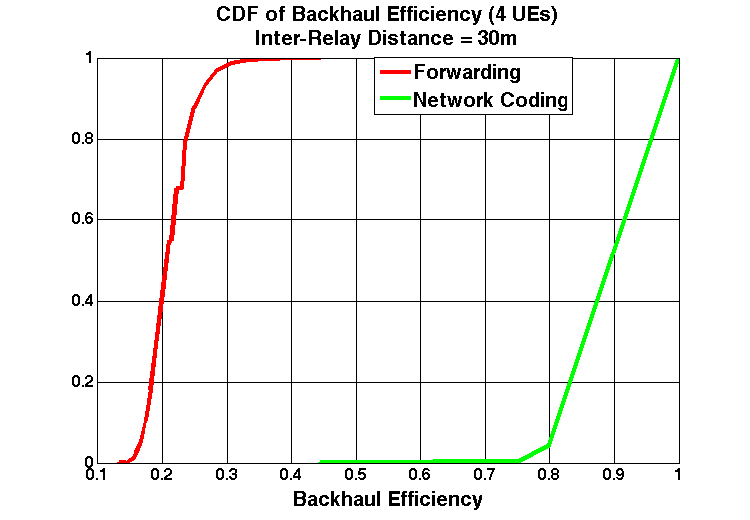}}
        \caption{\label{fig:Fig6e}}
      \end{subfigure}
    }
    &%
    { \begin{subfigure}[t]{0.47\textwidth}
        \fbox{\includegraphics[height=5.8cm, width=\textwidth]{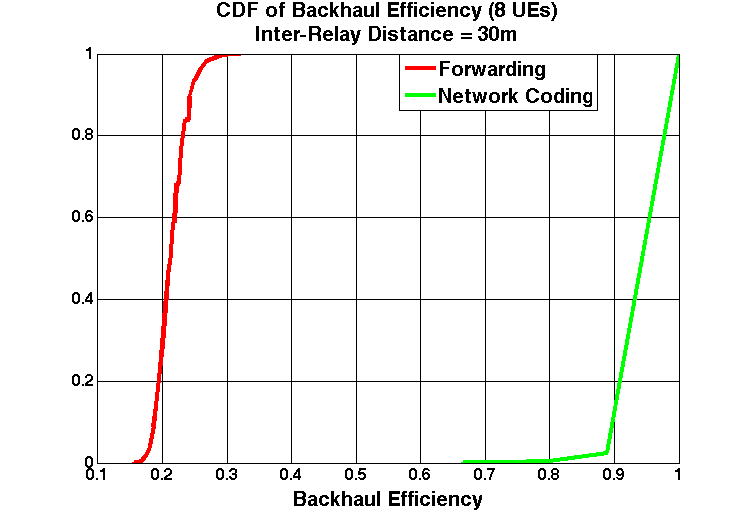}}
        \caption{\label{fig:Fig6f}}
      \end{subfigure}
    } \\%
    
  \end{tabular}
  \caption{Uplink backhaul efficiency network coding compared to
    forwarding for different network code length (number of devices)
    and relay densities}
\label{fig:Fig6}
\end{figure*}

In this section, Monte-Carlo system simulation results are provided to
evaluate the benefits of network coding compared to the forwarding
approach. It is assumed that $N_R$ relay nodes are placed uniformly
spaced on two sides of a street with width $w_s$. The distance between
the relays is $D_R$, and relay nodes on one side of the street are
shifted by $\frac{D_R}{2}$ with respect to the relay nodes on the
other side. The street has a sidewalk at 2 meters from each
edge. $N_{ue}$ user devices are uniformly randomly dropped on
sidewalks. Figure \ref{fig:Fig4} shows the simulation setup.

\subsection{Channel Model}
\label{subsec-ChanMod}
In this work, a packet erasure $p_{i,j}$ is associated with the link
between device $i$ and relay node $j$, and a packet is assumed to be
received correctly with probability $1-p_{i,j}$ at the
destination. The channel model illustrated in [19] is used to capture
fundamental characteristics of the millimeter-wave channels, namely
high path-loss, and frequent outage.  In particular, a link is assumed
to be in Outage, Line-of-Sight or Non-Line-of-Sight states with
certain device-relay node distance based probabilities, as described
in [19]. If a link is in outage, the corresponding packet erasure
probability is 1.  If the link is in line-of-sight or
non-line-of-sight states, the link erasure probability is set to the
block error rate (BLER) corresponding to that link. To obtain the
BLER, path-loss is computed according to the path-loss model for 28
GHz described in [19]. Then SNR is derived based on assumed parameter
values for beam-forming gain, transmit-power, coding gain, as well as
path-loss. Once SNR is calculated, BLER is determined assuming 64 QAM
transmissions for the downlink and QPSK transmissions for the uplink,
respectively. Further, if the erasure probability of a link is higher
than a threshold $T_e$, that link is not used for communication and
is assumed to be in outage. The parameter values used and other
details are listed in appendix \ref{sec: Appendix-Simparam}.

\subsection{Downlink Simulations}
\label{subsec-DownSim}
In this section, we discuss simulation results comparing the
Forwarding and Network Coding approaches for downlink transmissions to
the devices (i.e., intra-session network coding as described in
section \ref{subsec-intrasessNC}). For this simulation, 5000 devices
are randomly dropped across the deployment detailed in section
\ref{subsec-DepScn}. During each time span, we compare efficiency and
packet delay metrics for transmissions of $k$ packets using the
forwarding approach and the network coding approach, as described in
section \ref{subsec-sched}. Figure~\ref{fig:Fig5a},
Figure~\ref{fig:Fig5c} and Figure~\ref{fig:Fig5e} show the CDF of
efficiency calculated for each device for Forwarding and Network
Coding approaches for $D_R$ of 30 meters, 60 meters and 80 meters,
respectively. The gain of network coding compared to forwarding
increases with increasing density of relay nodes. For example, the
case with $D_R = 30$ meters shows a larger gain for network coding
compared to that of $D_R = 60$ meters and $D_R = 80$ meters. This
behavior is attributed to having more outage-free links at higher
relay node densities (although the links have varying error
rates). The improvement in efficiency points to an increase in data
rates of 35\%-39\% for network coding compared to forwarding, for
$D_R$ of 30 meters. The case with $D_R$ of 60 meters shows an increase
in data rates of 11\%-37\%. The case with $D_R$ of 80 meters shows an
increase in the range of 3\%-34\%.

Figure~\ref{fig:Fig5b}, Figure~\ref{fig:Fig5d} and
Figure~\ref{fig:Fig5f} show the delay in delivering the packets in
each time-span.  As mentioned previously, $k$ packets are delivered in
each time-span. The delay represents the number of time-slots used for
transmissions (i.e., transmission time) until all the $k$ packets are
received at the device. As shown, network coding yields significant
reductions in the delays. The gains are more significant for $D_R$ of
30 meters, i.e., the denser deployment, and decrease with decreasing
relay node density.

\subsection{Simulation of Backhaul Efficiency in Uplink}
\label{subsec-UpSim}

In this section, we discuss simulation results comparing the
Forwarding and Network Coding approaches for uplink transmissions
(i.e., inter-session network coding as described in section
\ref{subsec-intersess}) with $k=1$ to capture the multi-user nature of
the uplink problem. Figure~\ref{fig:Fig6a} and Figure~\ref{fig:Fig6b}
show a comparison of the backhaul efficiencies for the Forwarding and
the Network Coding approaches, when packets of 4 and 8 devices are
combined at relay nodes respectively, for $D_R = 80$
meters. Figure~\ref{fig:Fig6c} and Figure~\ref{fig:Fig6d}, and
Figure~\ref{fig:Fig6e} and Figure~\ref{fig:Fig6f} show corresponding
comparisons of the backhaul efficiencies for $D_R = 60$ meters and
$D_R = 30$ meters, respectively. As can be seen from
Figure~\ref{fig:Fig6}, the backhaul efficiency of network coding
improves relative to forwarding when either (a) the inter-relay
distance is reduced, or (b) the number of devices is increased. The
median network coding backhaul efficiency with 4 devices is
approximately 57\%, 100\% and 320\% better than the median forwarding
backhaul efficiency for inter-relay distances of 80, 60 and 30 meters,
respectively. With 8 devices, the median network coding backhaul
efficiency is approximately 63\%, 123\% and 352\% better than the
median forwarding backhaul efficiency for inter-relay distance of 80,
60 and 30 meters, respectively. The improvement in
backhaul efficiencies directly translate to reduction in backhaul
traffic - for example, the 100\% improvement in backhaul efficiency
due to network coding with 4 devices for the $D_R = 60$ meter case
corresponds to a 50\% reduction in backhaul traffic.

\section{CONCLUSION}
\label{sec:Concl}
We have analyzed and quantified the benefits of Random Linear Network
Coding in millimeter-wave communication systems with a dense deployment
of access points, where individual links show a lot of variation. For
our analysis and comparisons, we use the ratio of number of packets
delivered to the number of transmissions needed for delivery of the
packets as an efficiency metric. For downlink communication, we focus
on the efficiency of network coding for the air-interface
transmissions, whereas for uplink communication, we focus on the
efficiency of network coding for the backhaul transmissions.

We have provided a theoretical upper bound for the expected values of
efficiency of downlink air-interface transmissions for Forwarding and
a theoretical lower bound for the expected values of efficiency for
Network Coding. For the scenarios of interest, we observe that the
lower bound for the expected value of efficiency of network coding is
higher than the upper bound of the expected value of efficiency for
forwarding. This implies that network coding is expected to outperform
forwarding, with the difference being dependent on the diversity of
packet erasure probabilities for the links.

We have derived a lower bound on the uplink backhaul efficiency of the
network coding for the case of user-to-relay links with equal packet
erasure probability. The bound is applicable to a range of erasure
probabilities (low to medium) and implies that network coding can
maintain a backhaul efficiency of close to 1 even for small network
code lengths. The range of erasure probabilities for which the bound
is applicable increases with increasing network code lengths. We also
provided a tight approximation for the backhaul efficiency for the
forwarding case. Comparing the lower bound of network coding to the
approximate backhaul efficiency of forwarding, we have shown that network
coding offers not only a much higher efficiency (multiple times
better) but also a universal backhaul efficiency close to one for the
applicable region of erasure probabilities.

For the simulation, we use a millimeter wave channel model along with
probabilities of links being in outage, non-line-of sight or
line-of-sight, to capture the rapidly changing nature of the
individual links. For the downlink intra-session network coding, our
results show a significant improvement in the efficiency of air-interface
transmissions with the use of network coding. The extent of
improvement depends on the density of the Relay nodes (as represented
by the parameter $D_R$). For example, with Relay nodes deployed every
30 m on each side of the street, we see a median improvement of 35\%
in efficiency, which in turn translates to correspondingly higher data
rates and shorter transmission durations.

For the uplink inter-session network coding, our results show a
significant improvement in the efficiency of backhaul transmissions
with the use of network coding. For example, with Relay nodes every 60
m on each side of the street, we see a median improvement in backhaul
efficiency of 100\% with 4 devices and 123\% with 8 devices (i.e. 50\%
and 55\% lower backhaul traffic for the same number of packets
respectively).

It is remarked that this study aimed to characterize
fundamental benefits of network coding in 5G wireless networks. To
achieve this goal a few assumptions about the physical layer
transmission procedures and protocols are made to simplify the system
model. Therefore, our observations and comparison metrics are useful
only in a relative sense. In order to understand the performance of
network coding in more absolute terms (e.g., spectral efficiency
improvement in b/s/Hz or latency reduction in milliseconds) it is
necessary to have a more comprehensive study using a detailed physical
layer channel structure. The details of such channel models will be
discussed in the coming months in standards bodies such as 3GPP.

\section{Appendix - Backhaul Efficiency Derivation for Network Coding}\label{sec: Appendix-NC-Backhaul}
\label{sec-ncBoundDerv}
In this appendix, the average number of required Backhaul
Transmissions (BHT) for network coding, $\beta_{nc}$, is obtained. The
averaging is done over source-relay link conditions when each
source-relay link erasure probability is $p$. Each backhaul
transmission from a relay is a network-coded packet constructed at the
relay from $z$ encoding coefficients. The $z$ encoding coefficients
are independently and identically distributed elements each drawn
according to distribution $\Omega$ (see section
\ref{subsec-intersess}). The network has to collect $z$ independent
network coded packets from the relays in order to decode all the $z$
original packets transmitted from the sources.

Assume that the network has received $z+l$ ($l\geq 0$) BHTs from the
relays. If it can construct a full rank $z\times z$ matrix from the
set of $z+l$ rows (each row corresponding to the encoding coefficient
vector for a network coded packet), the network can decode all the
original packets. We use $\mathbb{M}_l$ to denote the set of
$\zeta(l)=\binom{z+l}{z}$ possible $z\times z$ matrices that can be
constructed from the $z+l$ rows. Each matrix belonging to the set is
singular with probability $\phi$ which is a function of $z$, $p$, and
the field size $q$.

The number of required BHT is at most $z+l$ when every matrix $M^o$ in
the set $\mathbb{M}_{l-1}$ is singular and there is at least one
$z\times z$ matrix $M^*$ in the set $\mathbb{M}_l$ which is
non-singular. In the following, the rank of a matrix $M$ is
represented by $\nu(M)$, and $Pr(.)$ denotes the probability of an
event.
\begin{eqnarray}
\label{eq: BHT_bound}
&\beta_{nc}(\phi)&= \sum_{l=0}^{\infty} (z+l) \ Pr(z+l
\ \text{BHTs are needed})\nonumber\\ 
       {}&\leq &\sum_{l=0}^{\infty} (z+l)\ Pr\Big(\exists M^*\in
       \mathbb{M}_l \setminus \mathbb{M}_{l-1},\nonumber\\
       & &\ \hspace{0.5cm} \nu(M^*)=z, \nexists M^o\in
       \mathbb{M}_{l-1}, \nu(M^o)=z\Big)\nonumber\\ 
       &=&\sum_{l=0}^{\infty} (z+l)\,Pr\Big(\exists M^*\in
       \mathbb{M}_l \setminus \mathbb{M}_{l-1},\nu(M^*)=z \mid\nonumber\\
       & & \hspace{2.2cm} \nexists M^o\in \mathbb{M}_{l-1},
       \nu(M^o)=z\Big)\nonumber\\ 
       & &\hspace{0.5cm}\times\ Pr\Big(\nexists M^o\in \mathbb{M}_{l-1},
          \nu(M^o)=z\Big)\nonumber\\ 
          &=&\sum_{l=0}^{\infty} (z+l)\ Pr\Big(\nexists
          M^o\in \mathbb{M}_{l-1},
          \nu(M^o)=z\Big)\nonumber\\
          & &\times \ \Bigg[1-Pr\Big(\nexists M^*\in\mathbb{M}_l \setminus \mathbb{M}_{l-1},\nonumber\\
            & & \hspace{0.5cm}\nu(M^*)=z \mid \nexists M^o\in \mathbb{M}_{l-1},
            \nu(M^o)=z\Big)\Bigg]\nonumber\\ 
          &=&\sum_{l=0}^{\infty} (z+l)\ \phi^{\zeta(l-1)}
            \ \big[1-\phi^{\zeta(l)-\zeta(l-1)}\big]
\end{eqnarray} 

The following
lemma proves that for $\phi\leq \phi_{ub}$, we have
$\beta_{nc}(\phi)\leq \beta_{nc}(\phi_{ub})$.
\begin{lemma}
\label{lemma:Monotonicity}
$\beta_{nc}(\phi)$ is a non-decreasing function of $\phi$.
\end{lemma}
\begin{proof}
In the following, it is shown that $\frac{d\beta_{nc}}{d\phi}>0$ for
$0<\phi<1$.
\begin{eqnarray*}
\frac{d\beta_{nc}}{d\phi} &=&\frac{d }{d\phi} \big[\sum_{l=0}^{\infty}
  (z+l)[1 - \phi^{\zeta(l) - \zeta(l-1)}] \phi^{\zeta(l-1)}\big]\\ 
&=&\frac{d }{d\phi}\big[\sum_{l=0}^{\infty} (z+l)(\phi^{\zeta(l-1)} -
  \phi^{\zeta(l)})\big]\\ 
&=&\frac{d }{d\phi} \big[z(1-\phi) +
  \sum_{l=1}^{\infty} (z+l)(\phi^{\zeta(l-1)} -
  \phi^{\zeta(l)})\big]\\ 
&=& -z+\sum_{l=1}^{\infty} (z+l) \zeta(l-1)
   \phi^{\zeta(l-1) - 1}\\
   & & \hspace{0.2cm}-\ \sum_{l=1}^{\infty} (z+l) \zeta(l)\phi^{\zeta(l)-1}\\ 
&=&-z+\sum_{l=0}^{\infty} (z+l+1)\zeta(l) \phi^{\zeta(l)-1}\\
   & & \hspace{0.2cm}-\ \sum_{l=1}^{\infty} (z+l)\zeta(l)
   \phi^{\zeta(l)-1}\\ 
&=&-z+(z+1)\zeta(0)\phi^{\zeta(0)-1} +\\ 
& &\hspace{0.2cm}\sum_{l=1}^{\infty} (z+l+1)\zeta(l)\phi^{\zeta(l)-1}
-\sum_{l=1}^{\infty} (z+l)\zeta(l)\phi^{\zeta(l)-1}\\
&=&\sum_{l=0}^{\infty} \zeta(l) \phi^{\zeta(l)-1}>0.
\end{eqnarray*}\end{proof}

Next we derive an upper bound $\phi_{ub}$ for the probability $\phi$
of a $z\times z$ matrix with elements drawn according to distribution
$\Omega$. $d(M)$ is used to denote $z-\text{rank}(M)$,
called the \emph{defect} of $M$. First we restate the following
results from \cite{Blomer1997}.

\begin{definition}\label{def: ld}
Let $M$ be a $z\times z$ matrix over Galois field $\text{GF}[q]$. Let
$M_1,\cdots,M_z$ denote the rows of $M$. A vector $(c_1,\cdots,c_z)
\in (\text{GF}[q])^{z}$ such that not all $c_i$ are zero is called a
linear dependency of $M_1,\cdots,M_z$ iff $\sum_{i=1}^{z} c_i M_i = 0$.
\end{definition}

\begin{theorem}
(Lemma 3.2 and Theorem 3.3 of \cite{Blomer1997}) Let $M$ be a random $z\times z$
  matrix over a Galois field of size $q$, where elements of $M$ are
  chosen according to the probability distribution $\Omega$. Let
  $d(M)$ denote the defect of $M$ and $l(M)$ denote the number of
  linear dependencies of the rows of $M$. Then,
\begin{equation}
\label{eq:dM}
q^{d(M)} - 1 = l(M).
\end{equation}
The expected number of linear dependencies of the rows of $M$ is:
\begin{equation}
\label{eq:lm1}
\hspace{-3pt}\bar{l}(M)=\sum_{k=1}^{z}\binom{z}{k}
\frac{1}{q^{z-k}}(1-\frac{1}{q})^k\,
\Big[1+(q-1)\big(1-\frac{q(1-p)}{q-1}\big)^k\Big]^z.
\end{equation}  
\end{theorem}
Now we show that
\begin{lemma}
\label{lemma:phi_ub}
$\phi\leq\phi_{ub}$ where $\phi_{ub}\triangleq \log_q(\bar{l}(M)+1)$.
\end{lemma}
\begin{proof}
Let $\nu$ denote the rank of the random $z\times z$ matrix $M$. Let
$d(M)$ denote the defect of matrix $M$ (i.e., $z-\nu$) and $\bar{d}(M)$
denote the expected value of the defect of $M$. It is first shown that
for the random matrix $M$, $\phi\leq \bar{d}(M)$, and then
$\bar{d}(M)\leq \phi_{ub}$.
\begin{eqnarray}
\label{eq:gamma}
 \bar{\nu}&=&\sum_{k=1}^{z}k\,Pr(\nu=k) \nonumber\\
 &=&\sum_{k=1}^{z-1}k\,Pr(\nu=k)+z\,(1-\phi)\nonumber\\
&\leq &  (z-1)\,\sum_{k=1}^{z-1}Pr(\nu=k)+z\,(1-\phi)\nonumber\\
 & = & (z-1)\phi+z(1-\phi)\nonumber\\
 \bar{\nu}&\leq&z-\phi
\end{eqnarray}
From Eq (\ref{eq:gamma}), given that $\bar{d}(M) = z -\bar{\nu}$, we
have $\phi\leq\bar{d}(M)$.  
From Eq (\ref{eq:dM}) we have $d(M) = \log_{q}(l(M)+1)$. Therefore,
$\bar{d}(M) = \mathbb{E}[[\log_{q}(l(M)+1)]]$. Using the concavity
property of the logarithm function, we have 
\begin{equation}
\bar{d}(M) \leq \log_{q}(\bar{l}(M)+1).
\end{equation}
\end{proof}
Based on Lemmas \ref{lemma:Monotonicity}, and Eq (\ref{eq: BHT_bound}) we have:
\[\beta_{nc}(\phi) \leq \beta_{nc}(\phi_{ub}).\]
Substituting $\phi_{ub}$ from Lemma \ref{lemma:phi_ub}, we have:
\[\beta_{nc}(\phi) \leq \sum_{l=0}^{\infty} (z+l)\times
\big[1-\phi_{ub}^{\zeta(l)-\zeta(l-1)}\big] \times
\phi_{ub}^{\zeta(l-1)}.\]

\section{Appendix - Simulation Parameters}\label{sec: Appendix-Simparam}
\renewcommand{\arraystretch}{1.2}
\small
\begin{center}
  \begin{tabular}{|p{2.5cm}|p{5.2cm}|}
    \hline
        {\bf Parameters} & {\bf Values}\\ \hline
        $N_R$ (number of Relay Nodes) & 10\\ \hline
        $D_R$ (Inter-relay distance) & 30, 60 and 80 meters\\ \hline
        $w_s$ (Street Width) & 20 meters \\ \hline
        Path-loss & Cf. first row of table I\cite{Akdeniz2014}, 28 GHz parameters \\ \hline
        Outage, Line of Sight and Non-line of Sight probabilities & Cf. second row of table I and Eq (8)\cite{Akdeniz2014}\\ \hline 
        Transmit Power & 30 dBm (Downlink) and 20 dBm (Uplink) \\ \hline
        Beam-forming Gain & 20 dB (Downlink) and 0 dB (Uplink) \\ \hline
        Coding Gain & 6 dB \\ \hline
        Noise Power & -87 dBm \\ \hline
        Noise Figure & 5 dB \\ \hline
        Modulation & 64 QAM (Downlink) and QPSK (Uplink) \\ \hline
        BLER formula & 
        \begin{tabular}{l}
          $(1 - p_b)^{\text{block-length}}$, where \\
          $p_{b,64QAM} = 0.2917 \mathit{erfc}(\sqrt{\frac{9 SNR}{63}})$, and\\
          $p_{b,QPSK} = 0.5 \mathit{erfc}(\sqrt{SNR})$
        \end{tabular} \\ \hline
        $\text{block-length}$ & 10000 bits \\ \hline
        $T_e$ & 0.9 \\ \hline
  \end{tabular}
\end{center}
\end{document}